\documentclass[10pt]{article}

\usepackage{ragged2e}
\usepackage{geometry}
\makeatletter
\renewcommand\normalsize{%
   \@setfontsize\normalsize\@xpt{14}%
   \abovedisplayskip 10\p@ \@plus2\p@ \@minus5\p@
   \abovedisplayshortskip \z@ \@plus3\p@
   \belowdisplayshortskip 6\p@ \@plus3\p@ \@minus3\p@
   \belowdisplayskip \abovedisplayskip
   \let\@listi\@listI}
\normalsize
\renewcommand\small{%
   \@setfontsize\small\@ixpt{12}%
   \abovedisplayskip 8.5\p@ \@plus3\p@ \@minus4\p@
   \abovedisplayshortskip \z@ \@plus2\p@
   \belowdisplayshortskip 4\p@ \@plus2\p@ \@minus2\p@
   \def\@listi{\leftmargin\leftmargini
               \topsep 4\p@ \@plus2\p@ \@minus2\p@
               \parsep 2\p@ \@plus\p@ \@minus\p@
               \itemsep \parsep}%
   \belowdisplayskip \abovedisplayskip
}
\renewcommand\footnotesize{%
   \@setfontsize\footnotesize\@viiipt{10}%
   \abovedisplayskip 6\p@ \@plus2\p@ \@minus4\p@
   \abovedisplayshortskip \z@ \@plus\p@
   \belowdisplayshortskip 3\p@ \@plus\p@ \@minus2\p@
   \def\@listi{\leftmargin\leftmargini
               \topsep 3\p@ \@plus\p@ \@minus\p@
               \parsep 2\p@ \@plus\p@ \@minus\p@
               \itemsep \parsep}%
   \belowdisplayskip \abovedisplayskip
}
\renewcommand\scriptsize{\@setfontsize\scriptsize\@viipt\@viiipt}
\renewcommand\tiny{\@setfontsize\tiny\@vpt\@vipt}
\renewcommand\large{\@setfontsize\large\@xipt{15}}
\renewcommand\Large{\@setfontsize\Large\@xiipt{16}}
\renewcommand\LARGE{\@setfontsize\LARGE\@xivpt{18}}
\renewcommand\huge{\@setfontsize\huge\@xxpt{30}}
\renewcommand\Huge{\@setfontsize\Huge{24}{36}}
\makeatother

\usepackage[tableposition=top]{caption}
\usepackage{ifpdf}
\usepackage{xcolor}
\usepackage{paralist}
\usepackage{tabularx}
\usepackage{booktabs}
\usepackage{multirow}
\usepackage{graphicx}
\usepackage{sidecap}
\usepackage{natbib}
\renewcommand{\cite}{\citep}
\usepackage{amssymb}
\usepackage{amsmath}
\usepackage{mathdots}
\usepackage{lscape}
\makeatletter
\def\vcdots{\vbox{\baselineskip4\p@ \lineskiplimit\z@
    \kern3\p@\hbox{.}\hbox{.}\hbox{.}\kern3\p@}}
\makeatother
\usepackage{caption}
\usepackage{subfigure}
\usepackage{array}
\usepackage{url}
\usepackage{float}

\ifpdf
  \usepackage{epstopdf}
	\usepackage{microtype}
\fi
\graphicspath{{./}{./figures/}}

\usepackage{algorithm}
\usepackage[noend]{algorithmic}

\colorlet{TufteRed}{red!80!black}

\definecolor{theblue}{RGB}{0,0,180}
\colorlet{thered}{TufteRed}

\usepackage{dgleich-math}
\renewcommand{\mat}{\boldsymbol}

\ifpdf%
  \RequirePackage[colorlinks,pdfdisplaydoctitle
      ,citecolor=theblue
      ,linkcolor=theblue
      ,urlcolor=theblue
      ,hyperfootnotes=false,pdftex,pagebackref]{hyperref}%
\else
  \RequirePackage[colorlinks,pdfdisplaydoctitle
      ,citecolor=theblue
      ,linkcolor=theblue
      ,urlcolor=theblue
      ,hyperfootnotes=false,dvipdfm,pagebackref]{hyperref}
\fi
\renewcommand*{\backref}[1]{}
\renewcommand*{\backrefalt}[4]{%
  \ifcase #1 %
    No citations.%
  \or
    Cited on page #2.%
  \else
    Cited on pages #2.%
  \fi
}

\newcommand{\secref}[1]{\hyperref[sec:#1]{section~\ref*{sec:#1}}}
\newcommand{\Secref}[1]{\hyperref[sec:#1]{Section~\ref*{sec:#1}}}
\newcommand{\secrefp}[1]{\hyperref[sec:#1]{(section~\ref*{sec:#1})}}

\newcommand{\thmref}[1]{\hyperref[thm:#1]{theorem~\ref*{thm:#1}}}
\newcommand{\Thmref}[1]{\hyperref[thm:#1]{Theorem~\ref*{thm:#1}}}
\newcommand{\thmrefp}[1]{\hyperref[thm:#1]{(theorem~\ref*{thm:#1})}}

\newcommand{\figrefp}[1]{\hyperref[fig:#1]{(figure~\ref*{fig:#1})}}
\newcommand{\figref}[1]{\hyperref[fig:#1]{figure~\ref*{fig:#1}}}
\newcommand{\Figref}[1]{\hyperref[fig:#1]{Figure~\ref*{fig:#1}}}
\usepackage{transparent}

\newcommand{\vvbar}{\bar{\vv}}
\let\Re\relax 
\DeclareMathOperator*{\Re}{Re}
\newcommand{\Real}{\Re\BracesOf}

\renewcommand{\abstract}{}

\title{A Dynamical System for PageRank \\with Time-Dependent Teleportation}
\author{David F.~Gleich \and Ryan A.~Rossi
\and
Purdue University\\
Department of Computer Science\\
305 N. University St., West Lafayette, IN 47906\\
\{dgleich, rrossi\}@purdue.edu
}

\begin{document}

\maketitle

\begin{abstract}
We propose a dynamical system that captures changes to the network centrality of nodes as external interest in those nodes vary.  We derive this system by adding time-dependent teleportation to the PageRank score.  The result is not a single set of importance scores, but rather a time-dependent set.  These can be converted into ranked lists in a variety of ways, for instance, by taking the largest change in the importance score.  For an interesting class of the dynamic teleportation functions, we derive closed form solutions for the dynamic PageRank vector.  The magnitude of the deviation from a static PageRank vector is given by a PageRank problem with complex-valued teleportation parameters.  Moreover, these dynamical systems are easy to evaluate.  We demonstrate the utility of dynamic teleportation on both the article graph of Wikipedia, where the external interest information is given by the number of hourly visitors to each page, and the Twitter social network, where external interest is the number of tweets per month.  For these problems, we show that using information from the dynamical system helps improve a prediction task and identify trends in the data.
\end{abstract}


\section{Introduction}
\label{sec:intro}

The PageRank vector of a directed graph is the stationary distribution of a Markovian \emph{random surfer}.  At a node, the random surfer either
\begin{compactenum}
\item transitions to a new node from the set of out-edges, or
\item does something else (e.g.~leaves the graph and then randomly returns)~\cite{page1998pagerank,langville2006-book}.
\end{compactenum}
The probability that the surfer performs the first action is known as the damping parameter in PageRank, denoted $\alpha$.  The second action is called teleporting and is modeled by the surfer picking a node at random according to a distribution called the teleportation distribution vector or personalization vector. This PageRank Markov chain always has a unique stationary distribution for any $0 \le \alpha < 1$. In this paper, we focus on the teleportation distribution vector $\vv$ and study how changing teleportation behavior manifests itself in a dynamical system formulation of PageRank.

To proceed further, we need to formalize the PageRank model.  Let $\mA$ be the adjacency matrix for a graph where
$A_{i,j}$ denotes an edge from node $i$ to node $j$.  To avoid a proliferation of transposes, we define $\mP$
as the transposed transition matrix for a random-walk on
a graph: 
\[ P_{j,i} = \text{ probability of transitioning from node $i$ to node $j$. } \]
Hence, the matrix $\mP$ is \emph{column-stochastic} instead of
row-stochastic, which is the standard in probability theory.
Throughout this manuscript, we utilize uniform random-walks on a graph, 
in which case $\mP = \mA^T \mD^{-1}$ where $\mD$ is a diagonal matrix with
the degree of each node on the diagonal.  However, none of the theory
is restricted to this type of random walk and any 
column-stochastic matrix will do.  If any nodes have no out-links, we assume that they are adjusted in one of the standard ways~\cite{boldi2007-traps}.  Let $\vv$ be a teleportation distribution vector such that $v_i \ge 0$ and $\sum_i v_i = 1$.  This vector models where the surfer will transition when ``doing something else.''  The PageRank Markov chain then has the transition matrix: 
\[ \alpha \mP  + (1-\alpha) \vv \ve^T. \]
While finding the stationary distribution of a Markov chain usually involves computing an eigenvector or solving a singular linear system, the PageRank chain has a particularly simple form for the stationary distribution vector $\vx$:
\[ ( \mI - \alpha \mP) \vx = (1-\alpha) \vv. \]

Sensitivity of PageRank with respect to $\vv$ is fairly well understood.  \Citet{langville2006-book} devote a section to determining the Jacobian of the PageRank vector with respect to $\vv$.  The choice of $\vv$ is often best guided by an application specific measure.  By setting $\vv = \ve_i$, that is, the $i$th canonical basis vector: 
\[ \ve_i = \left[ \begin{smallmatrix} 0 \\ \vcdots \\ 0 \\ 1 \\ 0 \\ \vcdots \\ 0 \end{smallmatrix} \right] \begin{smallmatrix} \\ \phantom{\vcdots} \\ \\ \text{$i$th row,} \\ \\ \phantom{\vcdots}\\ \end{smallmatrix} \]  PageRank computes a highly localized diffusion that is known to produce empirically meaningful clusters and theoretically supported clusters~\cite{andersen2006-local,tong2006-random-walk-restart}.  By choosing $\vv$ based on a set of known-to-be-interesting nodes, PageRank will compute an expanded set of interesting nodes~\cite{gyongyi2004-trustrank,singh2007-matching-topology}.  Yet, in all of these cases, $\vv$ is chosen once for the graph application or particular problem.  

In the original motivation of PageRank~\cite{page1998pagerank}, the distribution $\vv$ should model how users behave on the web when they don't click a link.  When this intuition is applied to a site like Wikipedia, this suggests that the teleportation function should vary as particular topics become interesting.  For instance, in our experiments \secrefp{eval}, we examine the number of page views for each Wikipedia article during a period where a major earthquake occurred.  Suddenly, page views to earthquake spike -- presumably as people are searching for that phrase.  We wish to \emph{include} this behavior into our PageRank model to understand what is now important in light of a radically different behavior.  One option would be to recompute a new PageRank vector given the observed teleporting behavior at the current time.  Our proposal for a dynamical system is another alternative.  That is, we define a new model where teleportation is the time-dependent function:
\[ \vv(t). \]
At each time $t$, $\vv(t)$ is a probability distribution of where the random walk teleports.  Figure~\ref{fig:dpr-model} illustrates this model.  We return to a comparison between this approach and solving PageRank systems in \secref{related-work}.

\begin{figure}
\centering
\includegraphics{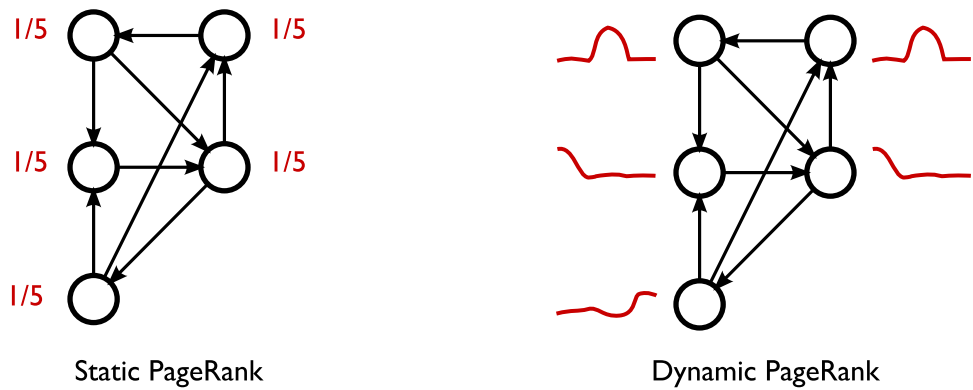}
\caption{At left, we have PageRank with static teleportation.  At each step, the teleportation is to each node with uniform probability $1/5$.  At right, we have the PageRank model with dynamic teleportation.  In this case, the teleportation distribution (illustrated in red) \emph{changes with time}. Thus, the upper nodes are teleported to more during the middle time regime.  In both cases, the graph is fixed.  We study the effect of such dynamic teleportation on the PageRank scores.}
\label{fig:dpr-model}
\end{figure}

The dynamical system we propose is a generalization of PageRank in the sense that if $\vv(t)$ is a constant function in time, then we converge to the standard PageRank vector \thmrefp{generalization}.  Additionally, we can analyze the dynamical PageRank function for some simple oscillatory teleportation functions $\vv(t)$.  Bounding the deviation of these oscillatory PageRank values from the static PageRank vector involves solving a \emph{PageRank problem with complex teleportation}~\cite{horn2008-parametric-google,constantine2010-rapr}.  This result is, perhaps, the first non-analytical use of PageRank with a complex teleportation parameter.

In our new dynamical system, we do not compute a single ranking vector as others have done with time-dependent rankings~\cite{grindrod2011communicability}, rather we compute a time-dependent ranking function $\vx(t)$, the dynamic PageRank vector at time $t$, from which we can extract different static rankings \secrefp{rank-to-time}. There are two complications with using empirically measured data.  First, we must choose a time-scale for our ODE based on the period of our page-view data \secrefp{time-scale}.  Put a bit informally, we must pick the time-unit for our ODE -- it is not dimensionless. We analytically show that some choices of the time-scale amount to solving the PageRank system for each change in the teleportation vector.  Second, we also investigate smoothing the measured page view data \secrefp{smoothing}.  To compute this dependent ranking function $\vx(t)$ we discuss ordinary differential equation (ODE) integrators in \secref{methods}.  

We discuss the impact of these choices on two problems: page views from Wikipedia and a retweet network from Twitter.  We also investigate how the rankings extracted from our methods differ from those extracted by other static ranking measurements.  We can use these rankings for a few interesting applications.  Adding the dynamic PageRank scores to a prediction task decreases the average error \secrefp{prediction} for Twitter.  Clustering the dynamic PageRank scores yields many of the standard time-series features in social networks~\secrefp{clustering}. Finally, using Granger causality testing on the dynamic PageRank scores helps us find a set of interesting links in the graph~\secrefp{granger}.

We make our code and data available in the spirit of reproducible research:

\centerline{ \footnotesize \url{http://www.cs.purdue.edu/homes/dgleich/codes/dynsyspr-im}}

\section{PageRank with time-dependent teleportation}\label{sec:dynamic-pagerank}
\label{sec:model}
\label{sec:dpr}

We begin our discussion by summarizing the notation introduced thus far
in Table~\ref{table:notation}.

\begin{table}[t!]
\caption{Summary of notation. Matrices are bold, upright roman letters; vectors are bold, lowercase roman letters; and scalars are unbolded roman or greek letters.  Indexed elements are vectors if they are bolded, or scalars if unbolded.}
\label{table:notation}
\centering 
\small
\begin{tabularx}{\linewidth}{rX} 
\toprule
$\imath$ & the imaginary number \\ 
$n$ & number of nodes in a graph  \\
$\ve$ & the vector of all ones \\
$\mP$ & column stochastic matrix  \\
$\alpha$ & damping parameter in PageRank  \\
$\vv$ & teleportation distribution vector \\
$\vx$ & solution to the PageRank computation: $(\eye - \alpha \mP) \vx = (1-\alpha) \vv$  \\
\midrule
$\vx(t)$ & solution to the dynamic PageRank computation for time $t$ \\
$\vv(t)$ & a teleportation distribution vector at time $t$  \\
\midrule
$\vc$ & the cumulative rank function \\
$\vr$ & the variance rank function \\
$\vd$ & the difference rank function \\
\midrule
$\vv_k$ & the teleportation distribution for the $k$th observed page-views vector \\
$\theta$ & decay parameter for time-series smoothing \\
$s$ & the time-scale of the dynamical system \\
$t_{\max}$ & the last time of the dynamical system \\
\bottomrule
\end{tabularx}
\end{table}
 
In order to incorporate changes in the teleportation into a new model for PageRank, we begin by reformulating the standard PageRank algorithm in terms of changes to the PageRank values for each page. This step allows us to state PageRank as a dynamical system, in which case we can easily incorporate changes into the vector.

The standard PageRank algorithm is the power method for the PageRank Markov chain~\cite{langville2006-book}.  After simplifying this iteration by assuming that $\ve^T \vx = 1$, it becomes:
\[ \vx\itn{k+1} = \alpha \mP \vx \itn{k} + (1-\alpha) \vv. \]
In fact, this iteration is equivalent to the Richardson iteration for the PageRank linear system $(\eye - \alpha \mP) \vx = (1-\alpha) \vv$.  This fact is relevant because the Richardson iteration is usually defined: 
\[ \vx\itn{k+1} = \vx\itn{x} + \omega \left[ (1-\alpha) \vv - (\eye - \alpha \mP) \vx\itn{k} \right]. \]
For $\omega = 1$, we have: 
\[ \Delta \vx\itn{k} = \vx\itn{k+1} - \vx\itn{k} = \alpha \mP \vx\itn{k} + (1-\alpha) \vv - \vx\itn{k} = (1-\alpha) \vv - (\eye - \alpha \mP) \vx\itn{k}. \]
Thus, changes in the PageRank values at a node \emph{evolve} based on the value $(1-\alpha) \vv - (\eye - \alpha \mP) \vx\itn{k}$.  We reinterpret this update as a continuous time dynamical system:
\begin{equation} 
 \vx'(t) = (1-\alpha) \vv - (\eye - \alpha \mP) \vx(t). 
\end{equation}
To define the PageRank problem with time-dependent teleportation, we make $\vv(t)$ a function of time.

\begin{definition}
\label{defn:dpr}
The dynamic PageRank model with time-dependent teleportation is the solution of 
\begin{equation} 
\label{eq:pr-dynamical}
 \vx'(t) = (1-\alpha) \vv(t) - (\eye - \alpha \mP) \vx(t)
\end{equation}
where $\vx(0)$ is a probability distribution vector and
$\vv(t)$ is a probability distribution vector for all $t$.
\end{definition}

In the dynamic PageRank model, the PageRank values $\vx(t)$ may not ``settle'' or converge to some fixed vector $\vx$.  We see this as a feature of the new model as we plan to utilize information from the evolution and changes in the PageRank values.  For instance, in \secref{rank-to-time}, we discuss various functions of $\vx(t)$ that define a rank.
Next, we state the solution of the problem.

\begin{lemma}
The solution of the dynamical system: 
\[ \vx'(t) = (1-\alpha) \vv(t) - (\eye - \alpha \mP) \vx(t) \]
is 
\[ \vx(t) = \exp[ -(\eye - \alpha \mP)t ] \vx(0) + (1-\alpha) \int_0^t \exp[-(\eye - \alpha \mP) (t-\tau)] \vv(\tau) \, d \tau. \]
\end{lemma}
This result is found in standard texts on dynamical systems, for example~\cite{Berman-1989-nonnegative}.

Given this solution, let us quickly verify a few properties of this system:
\begin{lemma}
The solution of a dynamical PageRank system $\vx(t)$ is a probability distribution ($\vx(t) \ge 0$ and $\ve^T \vx(t) = 1$) for all $t$. 
\end{lemma}
\begin{proof}
 The model requires that $\vx(0)$ is a probability distribution.  Thus, $\vx(0) \ge 0$ and $\ve^T \vx(0) = 1$.
Assuming that the sum of $\vx(t)$ is $1$, then the sum of the derivative $\vx'(t)$ is $0$ as a quick calculation shows.
The closed form solution above is also nonnegative because the matrix $\exp[-(\eye - \alpha \mP)] = \exp(\alpha \mP) \exp(-1) \ge 0$ and both $\vx(0)$ and $\vv(t)$ are non-negative for all $t$.  (This property is known as exponential non-negativity and it is another property of $M$-matrices such as $\eye - \alpha \mP$~\cite{Berman-1989-nonnegative}.)
\end{proof}

\subsection{A generalization of PageRank}
\label{sec:generalization}

This closed form solution can be used to solve a version of the dynamic problem that reduces to the PageRank problem with static teleportation.
If $\vv(t) = \vv$ is constant with respect to time, then
\[ \int_0^t \exp[-(\eye - \alpha \mP) (t-\tau)] \vv(\tau) \, d \tau = (\eye - \alpha \mP)^{-1} \vv - \exp[ -(\eye - \alpha \mP)t ] (\eye - \alpha \mP)^{-1} \vv. \]
Hence, for constant $\vv(t)$:
\[
\vx(t) = \exp[ -(\eye - \alpha \mP)t ] (\vx(0) - \vx) + \vx,
\]
where $\vx$ is the solution to static PageRank: $(\eye - \alpha \mP) \vx = (1-\alpha)\vv$.
Because all the eigenvalues of $-(\eye - \alpha \mP)$ are less than $0$, the matrix exponential terms disappear in a sufficiently long time horizon.
Thus, when $\vv(t) = \vv$, nothing has changed. We recover the original PageRank vector $\vx$ as the steady-state solution:
\[ \lim_{t \to \infty} \vx(t) = \vx \text{ the PageRank vector. } \]
This derivation shows that dynamic teleportation PageRank is a generalization of the PageRank vector.  We summarize this discussion as: 
\begin{theorem}
\label{thm:generalization}
PageRank with time-dependent teleportation is a generalization of PageRank. If $\vv(t) = \vv$, then
the solution of the ordinary differential equation: 
\[ \vx'(t) = (1-\alpha) \vv - (\eye - \alpha \mP) \vx(t) \]
converges to the PageRank vector 
\[ (\eye - \alpha \mP) \vx = (1-\alpha) \vv \]
as $t \to \infty$.
\end{theorem}

\subsection{Choosing the initial condition}

There are three natural choices for the initial condition $\vx(0)$.  The first choice is the uniform vector $\vx(0) = \frac{1}{n} \ve$.  The second choice is the initial teleportation vector $\vx(0) = \vv(0)$.  And the third choice is the solution of the PageRank problem for the initial teleportation vector $(\eye - \alpha \mP) \vx(0) = (1-\alpha) \vv(0)$.  We recommend either of the latter two choices in order to generalize the properties of PageRank.  Note that if $\vx(0)$ is chosen to solve the PageRank system for $\vv(0)$, then $\vx(t) = \vx$ for all $t$ is the solution of the PageRank dynamical system with constant teleportation~\thmrefp{generalization}.

\subsection{PageRank with fluctuating interest}

One of the advantages of the PageRank dynamical system is that we can study  problems analytically.  We now do so with the following teleportation function, or forcing function as it would be called in the dynamical systems literature:
\[ \vv(t) = \frac{1}{k} \sum_{j=1}^k \vv_j \Bigl(\cos (t + (j-1)\tfrac{2\pi}{k} ) + 1 \Bigr), \] 
where $\vv_j$ is a teleportation vector.  
Here, the idea is that $\vv_j$ represents the propensity of people to visit certain nodes at different times.  To be concrete, we might have $\vv_1$ correspond to news websites that are visited more frequently during the morning, $\vv_2$ correspond to websites visited at work, and $\vv_3$ correspond to websites visited during the evening.  
This function has all the required properties that we need to be a valid teleportation function.  With the risk of being overly formal, we'll state these as a lemma.

\begin{lemma}
Let $k \ge 2$.  Let $\vv_1, \ldots, \vv_k$ be probability distribution vectors. The time-dependent teleportation function 
\[ \vv(t) = \frac{1}{k} \sum_{j=1}^k \vv_j \Bigl(\cos (t + (j-1)\tfrac{2\pi}{k} ) + 1 \Bigr), \]
satisfies the both properties:
\begin{compactenum}
\item $ \vv(t) \ge 0 $ for all $t$, and
\item $ \sum_{i=1}^n v(t)_i = 1$ for all $t$
\end{compactenum}
\end{lemma}
\begin{proof}
The first property follows directly because the minimum value of the cosine function is $-1$, and thus, $\vv(t)$ is always non-negative.
The second property is also straightforward.  Note that 
\[ \sum_{i=1}^n v(t)_i = 1 + \sum_{j=1}^k \cos(t+(j-1)\tfrac{2\pi}{k}) = 1 + \sum_{j=1}^k \Real*{  \exp\bigl(\imath t + \imath (j-1) \tfrac{2\pi}{k} \bigr)}. \]
Let $r_j(t) = \exp(\imath t + \imath (j-1) \tfrac{2\pi}{k})$.  For $t=0$, these terms express the $k$ roots of unity. For any other $t$, we simply rotate these roots.  Thus we have $\sum_j r_j(t) = 0$ for any $t$ because the sum of the roots of unity is $0$ if $k \ge 2$.  The second property now follows because the sum of the real component is still zero.
\end{proof}

For this function, we can solve for the steady-state solution analytically.
\begin{lemma} \label{lem:fluctuate}
Let $k \ge 2$, $0 \le \alpha < 1$, $\mP$ be column-stochastic, $\vv_1, \ldots, \vv_k$ be probability distribution vectors, and \[ \vv(t) = \frac{1}{k} \sum_{j=1}^k \vv_j \Bigl(\cos (t + (j-1)\tfrac{2\pi}{k} ) + 1 \Bigr) = \frac{1}{k} \mV \cos(t + \vf) + \frac{1}{k} \mV \ve, \]
where $\mV = \bmat{ \vv_1, \ldots, \vv_k }$ and $f_j = (j-1) \tfrac{2\pi}{k} \quad j=1, \ldots, k$. Then the steady state solution of 
\[ \vx'(t) = (1-\alpha) \vv(t) - (\eye - \alpha \mP) \vx(t) \]
is 
\[ \vx(t) = \vx + \Real{\vs \exp( \imath t )} \]
where $\vx$ is the solution of the static PageRank problem 
\[ (\eye - \alpha \mP) \vx = (1-\alpha) \frac{1}{k} \mV \ve \]
and $\vs$ is the solution of the static PageRank problem with \emph{complex teleportation} 
\[ (\eye - \tfrac{\alpha}{1+\imath} \mP) \vs = (1-\alpha) \tfrac{1}{k(1+\imath)} \mV \exp(\imath \vf). \]
\end{lemma}
\begin{proof}
This proof is mostly a derivation of the expression for the solution by guessing the form.  First note that if 
\[ \vx(t) = \vx + \vy(t) \]
then 
\[ \vy'(t) = (1-\alpha) \tfrac{1}{k} \mV \cos(t + \vf) - (\eye - \alpha \mP) \vy(t). \]
That is, we've removed the constant term from the teleportation function by looking at solutions centered around the static PageRank solution.  To find the steady-state solution, we look at the complex-phasor problem: 
\[ \vz'(t) = (1-\alpha) \tfrac{1}{k} \mV \exp(\imath t + \imath \vf) - (\eye - \alpha \mP) \vz(t) \]
where $\vy(t) = \Real{\vz(t)}$.  Suppose that $\vz(t) = \vs \exp(\imath t)$.  Then: 
\[ \vz'(t) = \imath \vs \exp(\imath t) = (1-\alpha) \tfrac{1}{n} \exp(\imath t) \exp(\imath \vf) - (\eye - \alpha \mP) \vs \exp(\imath t). \]

The statement of $\vs$ in the theorem is exactly the solution after canceling the phasor $\exp(\imath t)$.
We now have to show that this solution is well-defined.  PageRank with a complex teleportation parameter $\gamma$ exists for any column-stochastic $\mP$ if $|\gamma| < 1$ (see \citet{horn2008-parametric-google,constantine2010-rapr}).  For the problem defining $\vs$, $\gamma = \alpha/(1+\imath)$ and $|\gamma| = \alpha/\sqrt{2}$.  Thus, such a vector $\vs$ always exists.
\end{proof}

We conclude with an example of this theorem.  Consider a four node graph with adjacency matrix and transition matrix: 
\[ \mA = \sbmat{     0    & 0    & 1    & 0 \\
    0    & 0    & 1    & 0 \\
    0    & 1    & 0    & 1 \\
    1    & 1    & 0    & 0} \qquad \text{ and } \qquad
\mP = \sbmat{ 0 & 0 & 0 & 0.5 \\ 0 & 0 & 0.5 & 0.5 \\ 1 & 1 & 0 & 0 \\ 0 & 0 & 0.5 & 0}. \]
Let $\vv_j = \ve_j$ for $j = 1, \ldots, 4$.  That is, interest oscillates between all four nodes in the graph in a regular fashion.  We show the evolution of the dynamical system for $20$ time-units in Figure~\ref{fig:cos-ff}.  This evolution quickly converges to the oscillators predicted by the lemma.  In the interest of simplifying the plot, we do not show the exact curves as they are visually indistinguishable from those plotted for $t \ge 4$.  By solving the complex valued PageRank to compute $\vs$, we can compute the magnitude of the fluctuation: 
\[ |\vs| = \bmat{    0.0216 & 0.0261 & 0.0122& 0.0235}^T. \]
This vector accurately captures the magnitude of these fluctuations.

\begin{SCfigure}[\sidecaptionrelwidth][h!]
\caption{The dashed lines represent the average PageRank vector computed for the teleportation vectors.  The curves show the evolution of the PageRank dynamic system for this example of teleportation.  We see that the dynamic PageRanks fluctuate about their average PageRank vectors.  Lemma~\ref{lem:fluctuate} predicts the magnitude of the fluctuation.}
\includegraphics[height=1.6in]{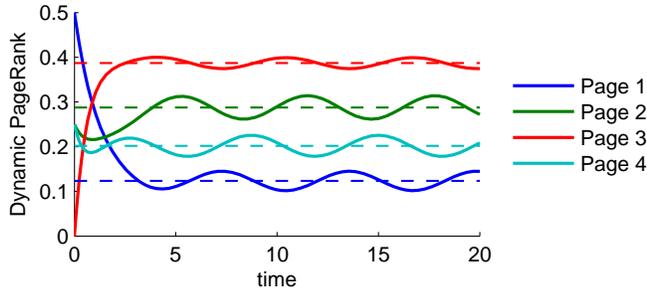}
\label{fig:cos-ff}
\end{SCfigure}

\subsection{Ranking from Time-Series} \label{sec:diff-ranking} \label{sec:rank-to-time}
The above equations provide a time-series of dynamic PageRank vectors for the nodes, denoted formally as $\vx(t), 0 \le t \le t_{\max}$.  Applications, however, often want a single score, or small set of scores, to characterize sets of interesting nodes.  There are a few ways in which these time series give rise to scores. 
Many of these methods were explained by \citet{o2005eventrank} in the context of ranking sequences of vectors.  Having a variety of different scores derived from the same data frequently helps when using these scores as features in a prediction or learning task~\cite{becchetti2008-spam,constantine2010-rapr}.

\paragraph{Transient Rank.} We call the instantaneous values of $\vx(t)$ a node's \emph{transient} rank.  This score gives the importance of a node at a particular time.

\paragraph{Summary, Variance, \textit{\&} Cumulative Rank.}  Any summary function $s$ of the time series, such as the integral, average, minimum, maximum, or variance, is a single score that encompasses the entire interval $[0,t_{\max}]$.  We utilize the \emph{cumulative rank, $\vc$} and \emph{variance rank, $\vr$} in the forthcoming experiments: 
\[ \vc = \int_0^{t_{\max}} \vx(t) \, dt \qquad \text{and} \qquad \vr = \int_0^{t_{\max}} (\vx(t) - \tfrac{1}{t_{\max}} \vc )^2 \, dt. \]

\paragraph{Difference Rank.} A node's difference rank is the difference between its maximum and minimum rank over all time, or a limited time window: 
\[ \vd = \max_t[\vx(t)] - \min_t[\vx(t)] \qquad \vd_W = \max_{t \in W} \vx(t) - \min_{t \in W} \vx(t). \]
Nodes with high difference rank should reflect important events that occurred within the range $[0, t_{\max}]$ or the time window $W$.  We suggest using a window $W$ that omits the initial convergence region of the evolution.  In the context of Figure~\ref{fig:cos-ff}, we'd set $W$ to be $[4,20]$ to approximate the vector $|\vs|$ numerically.  In Section~\ref{sec:results} and \figref{evolving-dpr-ts4}, we see examples of how current news stories arise as articles with high difference rank.

\subsection{Modeling activity}
\label{sec:time-scale}

In the next two sections of our introduction to the dynamic teleportation PageRank model, we discuss how to incorporate empirically measured activity into the model.  Let $\vp_1, \ldots, \vp_k$ be $k$ observed vectors of activity for a website.  In the cases we examine below, these activity vectors measure page views per hour on Wikipedia and the number of tweets per month on Twitter. We normalize each of them into teleportation distributions, and conceptually think of the collection of vectors as a matrix 
\[ \vv_1, \ldots, \vv_k \quad \to \quad \mV = \bmat{\vv_1, \ldots, \vv_k}. \]
Let $\ve(i)$ be a functional form representing the vector $\ve_i$. 
The time-dependent teleportation vector we create from this data is: 
\[ \vv(t) = \mV \ve(\text{floor}\BracesOf{t}+1) = \vv_{\text{floor}\BracesOf{t}+1}. \]
For this choice, the time-units of our dynamical system are given by the time-unit of the original measurements. Other choices are possible too.  Consider: 
\[ \vv_s(t) = \mV \ve(\text{floor}\BracesOf{t/s}+1) = \vv_{\text{floor}\BracesOf{t/s}+1}. \]
If $s > 1$, then time in the dynamical system slows down.  If $s < 1$, then time accelerates. Thus, we call $s$ the time-scale of the system.  Note that 
\[ \vx(sj), \qquad j = 0, 1, \ldots \]
represents the same effective time-point for any time-scale.  Thus, when we wish to compare different time-scales $s$, we examine the solution at such scaled points.

In the experimental evaluation, the parameter $s$ plays an important role.  We illustrate its effect in Figure~\ref{fig:param-examples}(a) for a small subnetwork extracted from Wikipedia. As we discuss further in \secref{methods}, for large values of $s$, then $\vv(t)$ looks constant for long periods of time, and hence $\vx(t)$ begins to converge to the PageRank vector for the current, and effectively static, teleportation vector.  Thus, we also plot the converged PageRank vectors as a step function.  We see that as $s$ increases, the lines converge to these step functions, but for $s=1$ and $s=2$, they behave differently.

\begin{figure*}[t!]
\centering
    \subfigure[timescale $s$]
    	{\label{fig:v1-ts}\includegraphics[width=0.5\linewidth]{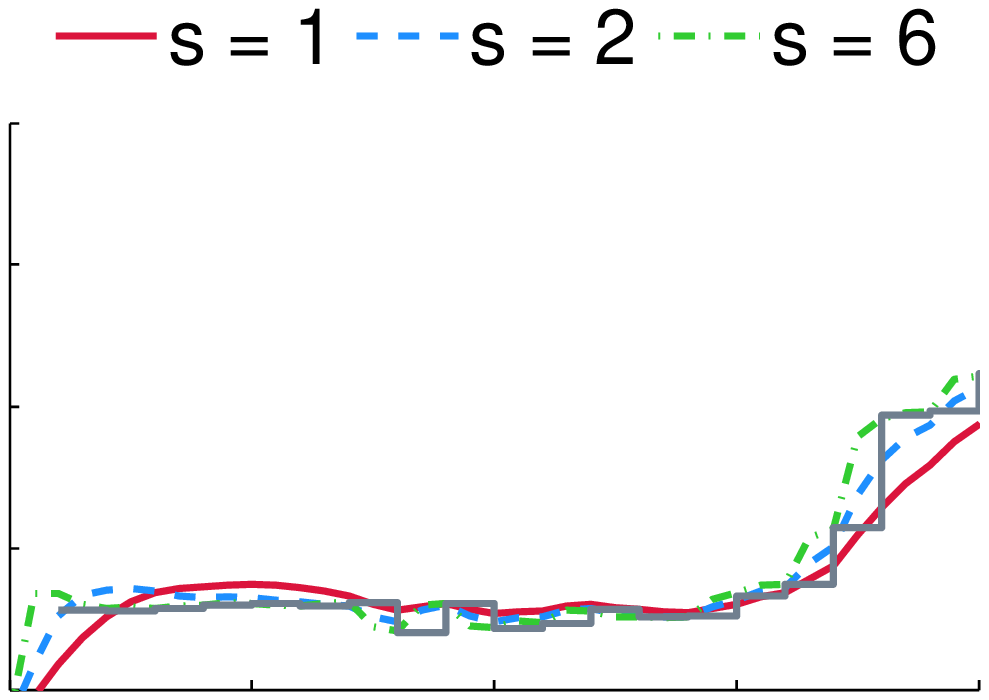}}%
    \subfigure[smoothing $\theta$]
    	{\label{fig:clusters}\includegraphics[width=0.5\linewidth]{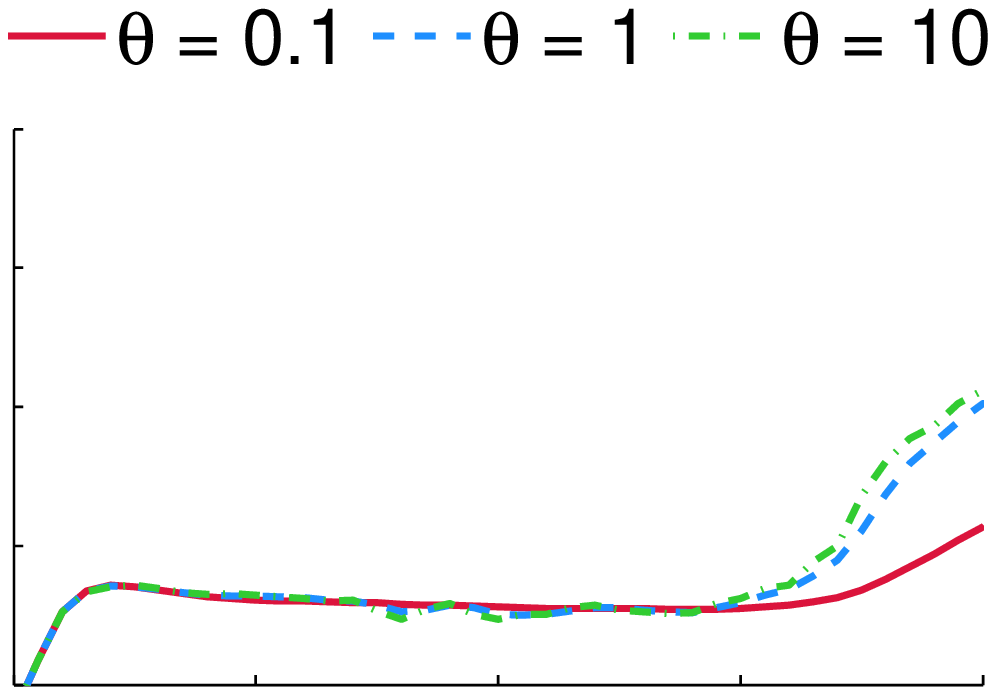}}
    \subfigure[damping parameter $\alpha$]
    	{\label{fig:centroid}\includegraphics[width=0.5\linewidth]{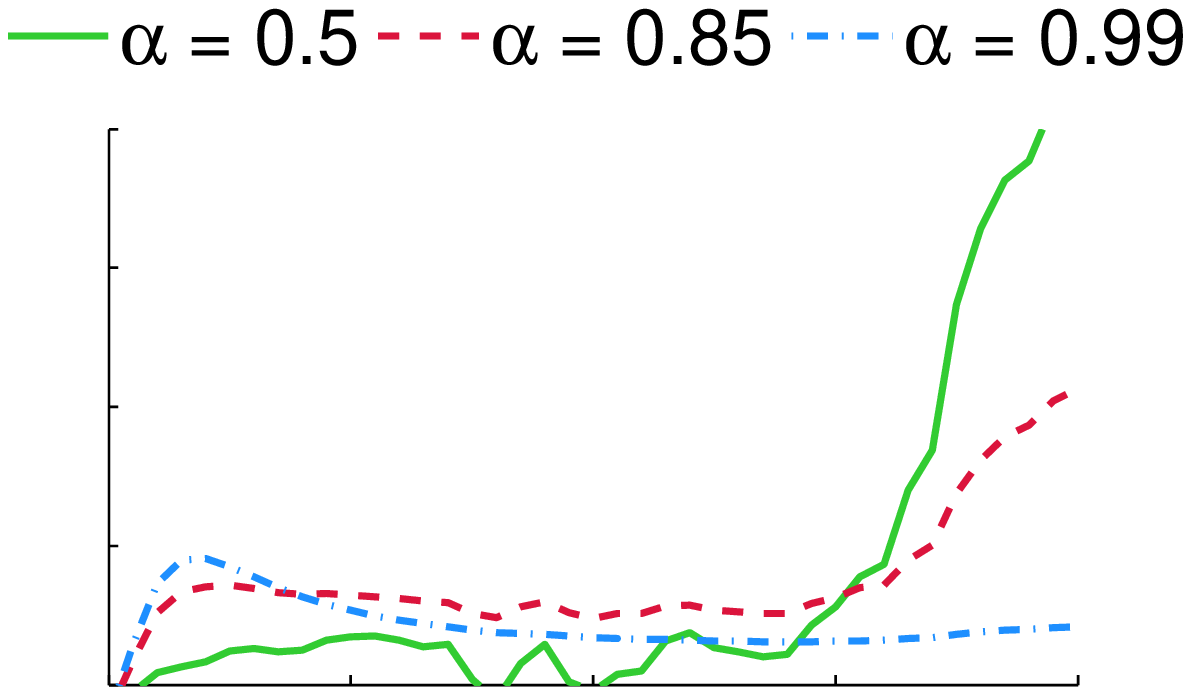}}
   \caption{The evolution of PageRank values for one node due to the dynamical teleportation.  The horizontal axis is time $[0,20]$, and the vertical axis runs between [0.01,0.014]. In figure (a), $\alpha=0.85$, and we vary the time-scale parameter \secrefp{time-scale} with no smoothing.  The solid dark line corresponds to the step function of solving PageRank exactly at each change in the teleportation vector.  All samples are taken from the same effective time-points as discussed in the section.  In figure (b), we vary the smoothing \secrefp{smoothing} of the teleportation vectors with $s = 2$, and $\alpha=0.85$.  In figure (c), we vary $\alpha$ with $s=2$ and no smoothing. We used the \texttt{ode45} function in Matlab, a Runge-Kutta method, to evolve the system.}
  \label{fig:param-examples}
\end{figure*}

\subsection{Smoothing empirical activity}
\label{sec:smoothing}

So far, we defined a time-dependent that $\vv(t)$ changes at fixed intervals based on empirically measured data.  A better idea is to smooth out these ``jumps'' using an exponentially weighted moving average. As a continuous time function, this yields:
\[ \vvbar'(t;\theta) = \theta \vv(t) - \theta \vvbar(t;\theta). \]
To understand why this smooths the sequence, consider an implicit Euler approximation: 
\[ \vvbar(t) = \frac{1}{1+h\theta} \vvbar(t-h;\theta) + \frac{h\theta}{1+h\theta} \vv(t).\]
This update can be written more simply as: 
\[ \vvbar(t;\theta) = \underbrace{\gamma \vv(t)}_{\text{new data}} + \underbrace{(1-\gamma)\vvbar(t-h;\theta)}_{\text{old data}}, \]
where $\gamma = \frac{h\theta}{1+h\theta}$.
When $\vv(t)$ changes at fixed intervals, then $\vvbar(t;\theta)$ will slowly change.  If $\theta$ is small then $\vvbar(t;\theta)$ changes slowly. We recover the ``jump'' changes in $\vv(t)$ in the limit $\theta \to \infty$.

The effect of $\theta$ is shown in Figure~\ref{fig:param-examples}(b).  Note that we quickly recover behavior that is effectively the same as using jumps in $\vv(t)$ ($\theta = 1, 10$).  So we only expect changes with smoothing for $\theta < 1$.

\subsection{Choosing the teleportation factor}
Picking $\alpha$ even for static PageRank problems is challenging, see~\citet{gleich2010tracking}~and~\citet{constantine2010-rapr} for some discussion.  In this manuscript, we do not perform any systematic study of the effects of $\alpha$ beyond Figure~\ref{fig:param-examples}(c).  This simple experiment shows one surprising feature.  Common wisdom for choosing $\alpha$ in the static case suggests that as $\alpha$ approaches 1, the vector becomes more sensitive.  For the dynamic teleportation setting, however, the opposite is true.  Small values of $\alpha$ produce solutions that more closely reflect the teleportation vector -- the quantity that is changing -- whereas large values of $\alpha$ reflect the graph structure, which is invariant with time.  Hence, with dynamic teleportation, using a small value of $\alpha$ is the sensitive setting.  Note that this observation is a straightforward conclusion from the equations of the dynamic vector: 
\[ \vx'(t) = (1-\alpha) \vv(t) + \alpha \mP \vx(t) - \vx(t) \]
so $\alpha$ small implies a larger change due to $\vv(t)$.  Nevertheless, we found it surprising in light of the existing literature.

\section{Methods for dynamic PageRank}
\label{sec:methods}

In order to compute the time-sequence of PageRank values $\vx(t)$, we can evolve the dynamical system \eqref{eq:pr-dynamical} using any standard method -- usually called an integrator. We discuss both the forward Euler method and a Runge-Kutta method next.  Both methods, and indeed, the vast majority of dynamical system integrators only require a means to evaluate the derivative of the system at a time $t$ given $\vx(t)$.  For PageRank with dynamic teleportation, this corresponds to computing: 
\[ \vx'(t) = (1-\alpha) \vv(t) - (\eye - \alpha \mP) \vx(t). \]
The dominant cost in evaluating $\vx'(t)$ is the matrix vector product $\mP \vx$.
For the explicit methods we explore, all of the other work is linear in the number of nodes, and hence, these methods easily scale to large networks.  
Both of these methods may also be used in a distributed setting if a distributed matrix-vector product is available.  

\subsection{Forward Euler}
We first discuss the forward Euler method.
This method lacks high accuracy, but is fast and straightforward.  
Forward Euler approximates the derivative with a first order Taylor approximation:
\[ \vx'(t) \approx \frac{\vx(t+h) - \vx(t)}{h}, \]
and then uses that approximation to estimate the value at a short time-step in the future: 
\[ \vx(t+h) = \vx(t) + h \left[ (1-\alpha) \vv(t) - (\eye - \alpha \mP) \vx(t) \right] . \]
This update is the original Richardson iteration with $h=\omega$.  We present the forward Euler method as a formal algorithm in Figure~\ref{fig:alg-dynpr} in order to highlight a comparison with the power and Richardson method.
That is, the forward Euler method is simply running a power method, but changing the vector $\vv$ at every iteration.  However, we derived this method based on evolving~\eqref{eq:pr-dynamical}.  Thus, by studying the relationship between~\eqref{eq:pr-dynamical} and the algorithm in Figure~\ref{fig:alg-dynpr}, we can understand the underlying problem solved by  changing the teleportation
vector while running the power method.  

\paragraph{Long time-scales.}
Using the forward Euler method, we can analyze the situation with a large time-scale parameter $s$.  Consider an arbitrary $\vx(0)$, $\alpha = 0.85$, $s=100$, $h=1$, and no smoothing.  In this case, then the forward Euler method will run the Richardson iteration for $100$ times before observing the change in $\vv(t)$ at $t=100$.   The difference between $\vx(k)$ and the exact PageRank solution for this temporarily static $\vv(t)$ is $\|\vx(k) - \vx\|_1 \le 2 \alpha^k$.
For $k > 50$, this difference is small. Thus, a large $s$ and no smoothing corresponds to solving the PageRank problem for each change in $\vv$.  

\paragraph{Stability.}
The forward Euler method with timestep $h$ is stable if the eigenvalues of the matrix $-h(\eye-\alpha \mP)$ are within distance $1$ of the point $-1$.  The eigenvalues of $\mP$ are all between $-1$ and $1$ because it is a stochastic matrix, and so this is stable for any $h < \frac{2}{1+\alpha}$.

\begin{figure}[t]
\centering
\begin{algorithmic}
 \REQUIRE ~\newline
  a graph $G=(V,E)$ and a procedure to compute $\mP \vx$ for this graph \newline
  a maximum time $t_{\max}$ \newline
  a function to return $\vv(t)$ for any $0 \le t \le t_{\max}$ \newline
  a damping parameter $\alpha$ \newline
  a time-step $h$
 \ENSURE $\mX$ where the $k$th column of $\mX$ is $\vx(0 + kh)$ 
  for all $1 \le k \le t_{\max}/h$ (or any
  desired subset of these values) \newline
 \STATE $t \leftarrow 0$; $k = 1$
 \STATE $\vx(0) \leftarrow \vv(0)$ (or any other desired initial condition)
 \WHILE {$t \le t_{\max}-h$}
   \STATE $\vx(t+h) \leftarrow \vx(t) + h \left[ (1-\alpha) \vv(t) - (\eye - \alpha \mP) \vx(t) \right] $
   \STATE $\mX(:,k) \leftarrow \vx(t+h)$
   \STATE $t \leftarrow t + h$; $k \leftarrow k + 1$
 \ENDWHILE  
\end{algorithmic}
\caption{The forward Euler method for evolving the dynamical system: $\vx'(t) = (1-\alpha) \vv(t) - (\eye - \alpha \mP) \vx(t)$. The resulting procedure looks remarkably similar to the standard Richardson iteration to compute a PageRank vector.  One key difference is that there is no notion of convergence.}\label{fig:alg-dynpr}
\end{figure}

\subsection{Runge-Kutta}

Runge-Kutta~\cite{runge1895numerische,kutta1901beitrag} numerical schemes are some of the most well-known and most used.
They achieve far greater accuracy than the simple forward Euler method, at the expense of a greater number of evaluations of the function $\vx'(t)$ at each step.
We use the implementations of Runge-Kutta methods available in the Matlab ODE suite~\cite{Shampine-1997-ODEs}.  The step-size is adapted automatically based on a local error estimate, and the solution can be evaluated at any desired point in time.  The stability region for Runge-Kutta includes the region for forward Euler, so these methods are stable.  These methods are also fast.  To integrate the system for Wikipedia with over 4 million vertices and 60 million edges, it took between 300-600 seconds, depending on the parameters.

\subsection{Maintaining interpretability}

Based on the theory of the dynamic teleportation system, we expect that $\vx(t) \ge 0$ and $\ve^T \vx(t) = 1$ for all time.  Although this property should be true of the computed solution, we often find that the sum diverges from one.  Consequently, for our experiments, we include a correcting term: 
\[ \vx'(t) = (1-\alpha) \vv(t) - (\gamma \eye - \alpha \mP) \vx(t) \]
where $\gamma =  (1-\alpha)\ve^T \vv(t) + \alpha \ve^T \vx(t)$.
Note that $\gamma = 1$ if the $\vx(t)$ has sum exactly $1$.  If $\ve^T \vx(t)$ is slightly different from one, then the correction with $\gamma$ ensures that $\ve^T \vx'(t) = 0$ numerically.  
Similar issues arise in computing static PageRank~\cite{wills2008-ordinal}, although the additional computation in the Runge-Kutta methods exacerbates the problem.

\section{Related work}\label{sec:related-work}
Note that we previously studied this idea in a conference
paper~\cite{Rossi-2012-dynamicpr}.  These ideas
have been significantly refined for this 
manuscript.

The relationship
between dynamical systems and classical iterative methods
has been utilized by~\citet{Embree-2009-dynamical}
to study eigenvalue solvers.  It was also noted in
an early paper by~\citet{Tsaparas-2004-dynamical-systems}
that there is a relationship between the PageRank and
HITS algorithms and dynamical systems.

In the past, others studied PageRank approximations
on graph streams~\cite{das2008estimating}.  More recently,
\citet{bahmani2012pagerank} studied how accurately
an evolving PageRank method could estimate the true
PageRank of an evolving graph that is accessed only
via a crawler.  The method used here solved each 
PageRank problem exactly for the current 
estimate of the underlying graph. A 
detailed study of how PageRank values evolve during
a web-crawl was done by \citet{boldi2005-incremental-pagerank}. 
In the future, we
plan to study dynamic graphs via similar ideas.

As explained in \secref{methods} and \figref{alg-dynpr}, 
our proposed method is related to
changing the teleportation vector in the power method
as its being computed.  Bianchini et 
al.~\cite{bianchini2005-inside-pagerank} noted
that the power method would still converge if
either the graph or the vector $\vv$ changed during
the method, albeit to a new solution given
by the new vector or graph.  Our method capitalizes
on a closely related idea and we utilize the 
intermediate quantities explicitly.  Another related
idea is the Online Page Importance Computation 
(OPIC)~\cite{abiteboul2003adaptive}, which integrates
a PageRank-like computation \emph{with} a crawling
process.  The method does nothing special if
a node has changed when it is crawled again.

While we described PageRank in terms of a random-surfer
model, another characterization of PageRank is that
it is a sum of damped transitions: 
\[ \vx = (1-\alpha) \sum_{k=0}^\infty (\alpha \mP)^k \vv. \]
These transitions are a type of probabilistic walk and 
Grindrod et al.~\cite{grindrod2011communicability}
introduced the related notion of dynamic walks for
dynamic graphs.  We can interpret these dynamic walks
as a \emph{backward Euler} approximation to the 
dynamical system: 
\[ \vx'(t) = \alpha \mA(t) \vx(t) \qquad \vx(0) = \ve \]
with time-step $h=1$ and $\mA$ is a time-dependent
adjacency matrix.  This relationship suggests that there may be
a range of interesting models between our dynamic teleportation
model and existing evolving graph models.

Outside
of the context of web-ranking, 
O'Madadhain and Smyth propose EventRank~\cite{o2005eventrank},
a method of ranking nodes in dynamic graphs, that
uses the PageRank propagation equations for a sequence of graphs.
We utilize the same idea but place it within the context
of a continuous dynamical system.  
In the context of popularity dynamics~\cite{ratkiewicz2010characterizing}, our method captures how changes in external interest influence the popularity of nodes and the nodes linked to these nodes in an implicit fashion.
Our work is also related to modeling human dynamics, namely, how humans change their behavior when exposed to rapidly changing or unfamiliar conditions~\cite{bagrow2011collective}. 
In one instance, our method shows the important topics and ideas relevant to humans before and after one of the largest Australian Earthquakes \figrefp{evolving-dpr-ts4}.

In closing, we wish to note that our proposed method \emph{does not involve}
updating the PageRank vector, a related problem which has
received considerable attention~\cite{chien2004link,langville04-iad}.
Nor is it related to tensor methods for dynamic graph
data~\cite{Sun-2006-tensor,Dunlavy-2011-temporal}.


\section{Examples of dynamic teleportation}
\label{sec:eval}
We now use dynamic teleportation to investigate page view patterns on Wikipedia and user activity on Twitter.  In the following experiments, unless otherwise noted, we set $s=1$, $\alpha = 0.85$, do not use smoothing (``$\theta = \infty$''), and use the \texttt{ode45} method from Matlab to evolve the system.  We study this model on two datasets.

\subsection{Datasets}\label{sec:datasets}
We provide some basic statistics of the Wikipedia and Twitter datasets in Table~\ref{table:dataset}.  For Wikipedia, the time unit for $s=1$ is an hour, and for Twitter, it is one month.  

\paragraph{Wikipedia Article Graph and Hourly page views.} Wikipedia provides access to copies of its database~\cite{wikipedia2009}. We downloaded a copy of its database on March 6th, 2009 and extracted an article-by-article link graph, where an article is a page in the main Wikipedia namespace, a category page, or a portal page. All other pages and links were removed. See~\citet{gleich2007three} for more information.

Wikipedia also provides hourly page views for each page~\cite{pageviews2009}. These are the number of times a page was viewed for a given hour. These are not unique visits. We downloaded the raw page counts and matched the corresponding page counts to the pages in the Wikipedia graph. We used the page counts starting from March 6, 2009 and moving forward in time.  Although it would seem like measuring page views would correspond to measuring $\vx(t)$ instead of $\vv(t)$, one of our earlier studies showed that users hardly ever follow links on Wikipedia~\cite{gleich2010tracking}.  Thus, we can interpret these page views as a reasonable measure of external interest in Wikipedia pages.

\paragraph{Twitter Social Network and Monthly Tweet Rates.}  We use a follower graph generated by starting with a few seed users and crawling follows links from 2008. We extract the user tweets over time from $2008 - 2009$. A tweet is represented as a tuple $\langle$user, time, tweet$\rangle$. Using the set of tweets, we construct a sequence of vectors to represents the number of tweets for a given month.

\begin{table}[h!]
\caption{Dataset Properties. The page views or tweets is denoted as $\vp$.}
\label{table:dataset}
\centering\small
\begin{tabularx}{\linewidth}{ l XX r@{\;\;} X XX } 
\toprule
Dataset &  Nodes & Edges & $t_{\max}$ & Period & Average $p_i$ & Max $p_i$  \\
\midrule
\textsc{wikipedia}  & 4,143,840   &  72,718,664  &  48  &  hours  &  1.4243  &  353,799\\
\textsc{twitter} & 465,022  & 835,424 & 6 &  months &  0.5569  & 1056  \\
\bottomrule
\end{tabularx}
\end{table}

\def \img_lsbig{2.6in}
\begin{figure*}[t!]
\centering
    \subfigure[In-degree]
    	{\label{fig:isim-indegree}\includegraphics[width=\img_lsbig]{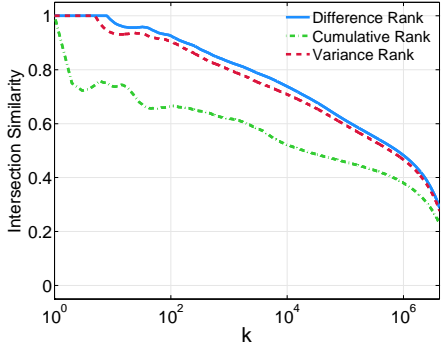}}
    \subfigure[Static PageRank (Uniform)]
    	{\label{fig:isim-pr-uniform}\includegraphics[width=\img_lsbig]{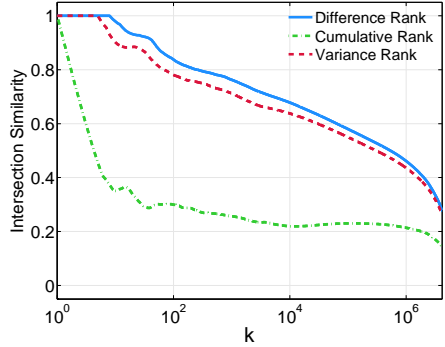}}
    \subfigure[Average page views]
    	{\label{fig:isim-avgpv}\includegraphics[width=\img_lsbig]{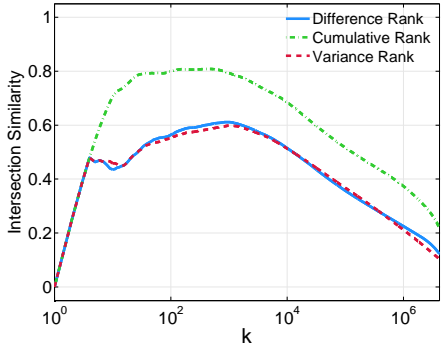}}
    \subfigure[Static PageRank (avg.~page views)]
    	{\label{fig:isim-pr-avgpv}\includegraphics[width=\img_lsbig]{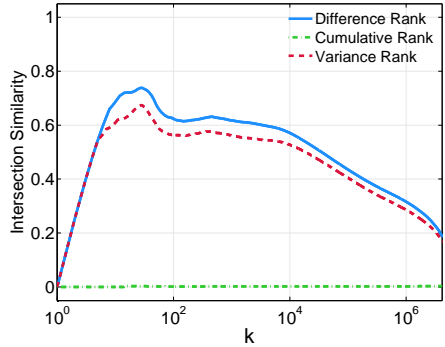}}
\caption{
Intersection similarity of rankings derived from dynamic PageRank. We compute the intersection similarity of the difference, variance, and cumulative rankings given by dynamic PageRank and compare these with the rankings given by the in-degree, average page views, static PageRank with uniform teleportation, and static PageRank with average page views as the teleportation vector.
For dynamic PageRank, we set the initial value $\vx(0)$ to be the solution of the static PageRank system which uses $\vv(0)$ as the teleportation vector. }
  \label{fig:isim}
\end{figure*}

\subsection{Rankings from transient scores}\label{sec:rank-ts}
First, we evaluate the rankings from dynamic PageRank using the intersection similarity measure~\cite{Boldi05totalrank}.
Given two vectors $\vx$ and $\vy$, the intersection similarity metric at $k$ is the average symmetric difference over the top-$j$ sets for each $j \leq k$. If $\mathcal{X}_k$ and $\mathcal{Y}_k$ are the top-$k$ sets for $\vx$ and $\vy$, then
\[ \mathrm{isim}_k(\vx,\vy) = \frac{1}{k} \sum_{j=1}^{k} \frac{|\mathcal{X}_j \Delta \mathcal{Y}_j|}{2j} \] 
\noindent 
where $\Delta$ is the symmetric set-difference operation. Identical vectors have an intersection similarity of 0.

For the Wikipedia graph, Figure~\ref{fig:isim} shows the similarity profile comparing a few ranking measures from dynamic PageRank to reasonable baselines.
In particular, we compare ${\vd}$, $\vr$, $\vc$  (from \S\ref{sec:diff-ranking}) to indegree, average page views, static PageRank with uniform teleportation, and static PageRank using average page views as the teleportation vector.
The results suggest that dynamic PageRank is different from the other measures, even for small values of $k$.  
In particular, combining the external influence with the graph appears to produce something new.
The only exception is in Fig.~\ref{fig:isim-pr-avgpv} where the cumulative rank is shown to give a similar ordering to static PageRank using average page views as the teleportation.

\subsection{Difference ranks} \label{sec:patterns}
\Figref{evolving-dpr-ts1}~and~\figref{evolving-dpr-ts4} show the time-series of the top 100 pages by the difference measure for Wikipedia with $s=1$ and $s=4$ without smoothing. Many of these pages reveal the ability of dynamic PageRank to mesh the network structure with changes in external interest.   For instance, in \figref{evolving-dpr-ts4}, we find pages related to an Australian earthquake (43, 84, 82), the ``recently'' released movie ``Watchmen'' (98, 23-24), a famous musician that died (2, 75), recent ``American Idol'' gossip (34, 63), a remembrance of Eve Carson from a contestant on ``American Idol'' (88, 96, 34), news about the murder of a Harry Potter actor (60), and the Skittles social media mishap (94).  These results demonstrate the effectiveness of the dynamic PageRank to identify interesting pages that pertain to external interest.  The influence of the graph results in the promotion of pages such as Richter magnitude (84).  That page was not in the top 200 from page views.

\begin{landscape}
\begin{figure*}
\includegraphics[width=\linewidth]{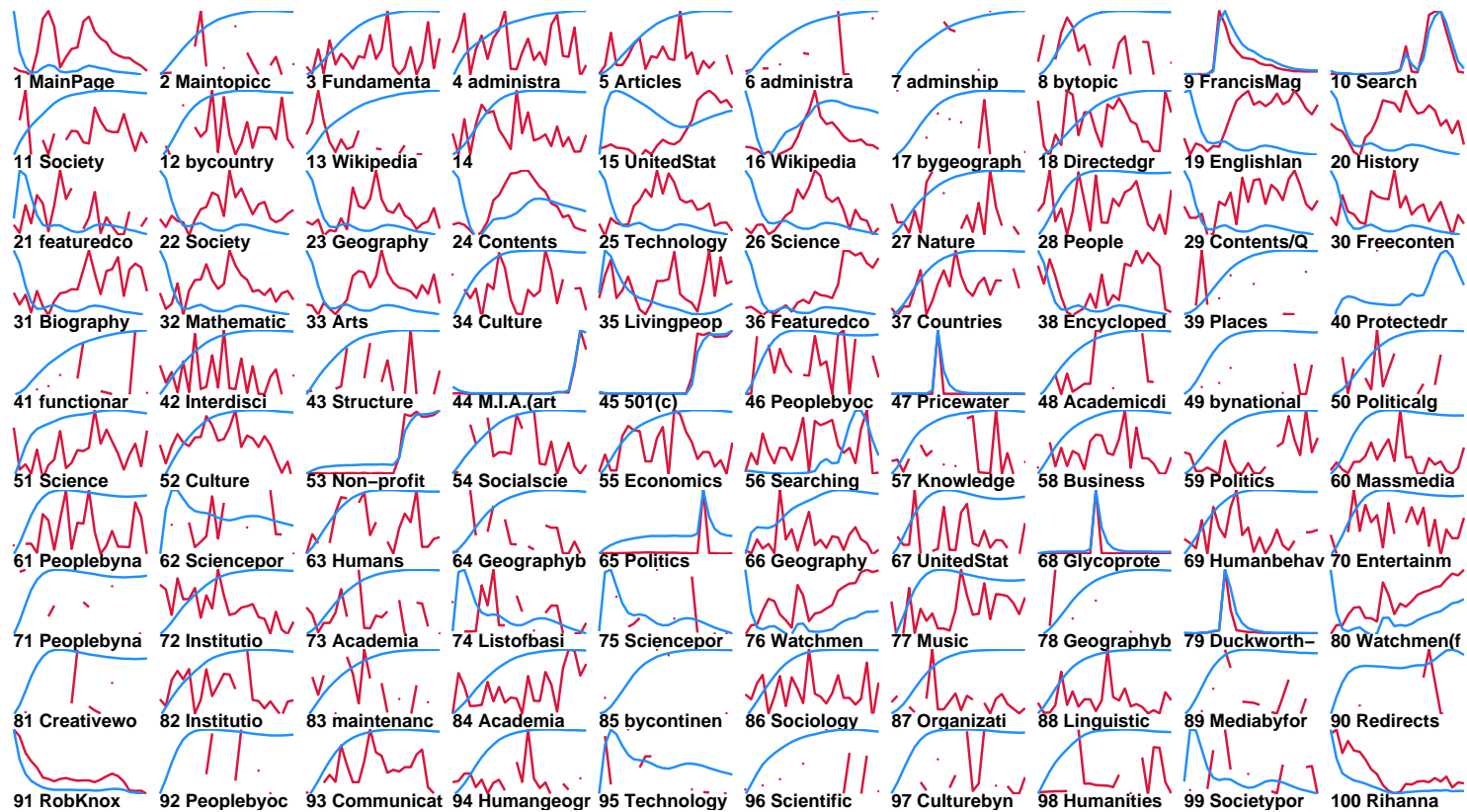}
   \caption{The top-100 Wikipedia pages that fluctuate the most as determined by the difference ranking from our dynamic PageRank approach. The blue curves are the transient scores, and the red-curves are the raw page view data.  The horizontal axis is 24 hours, and the vertical axes are normalized to show the range of the data.}
  \label{fig:evolving-dpr-ts1}
\end{figure*}

\begin{figure*}[t!]
\includegraphics[width=\linewidth]{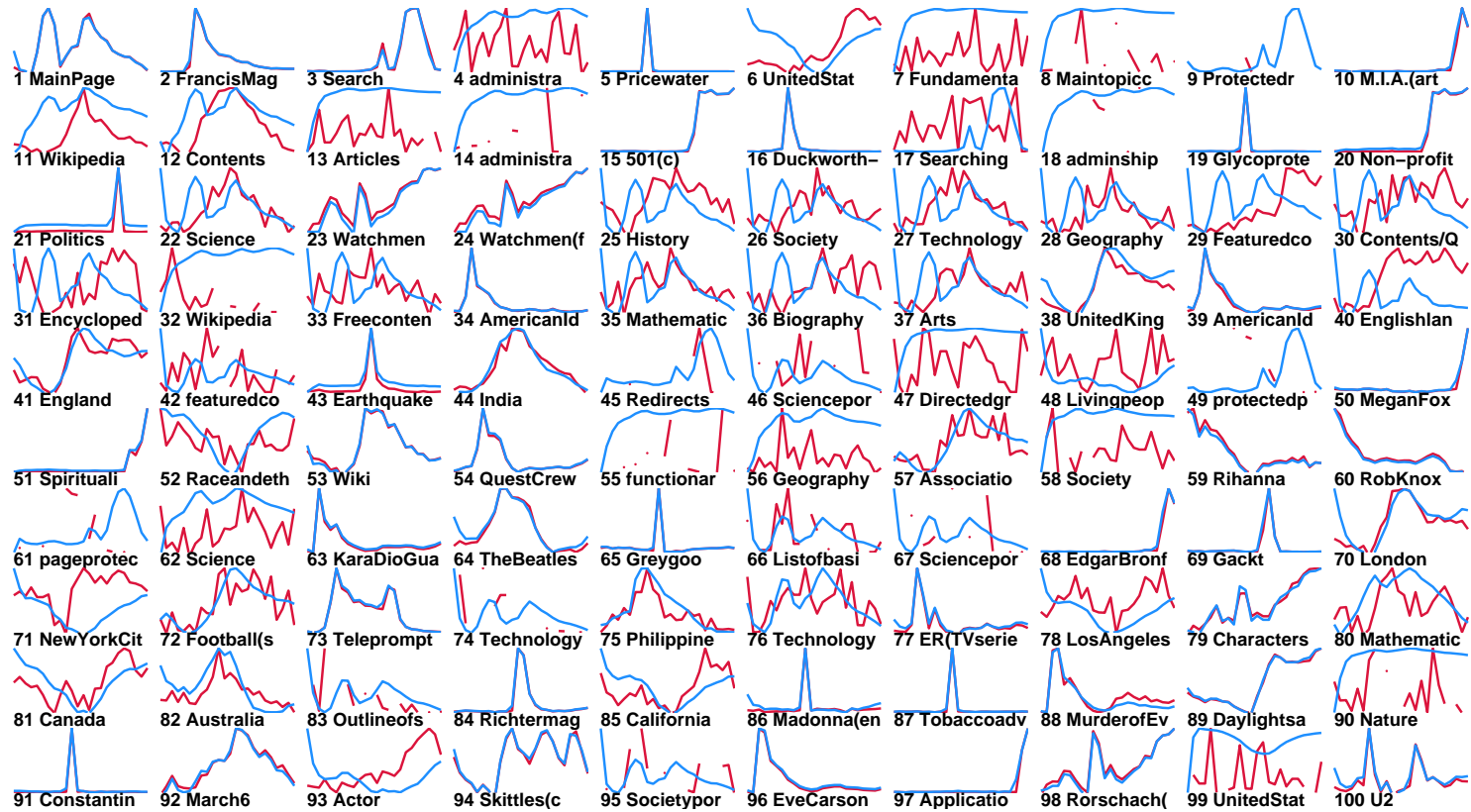}
   \caption{We again plot the top 100 pages from the difference rank, now using $s=4$.  The blue curves are the transient scores, and the red-curves are the raw page view data.  The horizontal axis is 24 hours, and the vertical axes are normalized to show the range of the data. This choice gives a similar ordering as our previous forward Euler iteration method from \citet{Rossi-2012-dynamicpr}. Note the large change between the transient scores for ''MainPage'' between this figure and \figref{evolving-dpr-ts1}.}
  \label{fig:evolving-dpr-ts4}
\end{figure*}

\end{landscape}

\section{Applications of time-dependent telportation}\label{sec:results}
This section explores the opportunity of using Dynamic PageRank for a variety of applications outside of the context of ranking.

\subsection{Predicting future page views \textit{\&} tweets}\label{sec:prediction}
We begin by studying how well the dynamical system can \emph{predict} the future.
Formally, given a lagged time-series $\vp_{t-w},..., \vp_{t-1}, \vp_{t}$ ~\cite{ahmed2010empirical}, the goal is to predict the future value $\vp_{t+1}$ (actual page views or number of tweets). 
This type of temporal prediction task has many applications, such as actively adapting caches in large database systems, or dynamically recommending pages. 

We performed one-step ahead predictions ($t+1$) using linear regression. 
That is, we learn a model of the form:
\[ \left[ \begin{array}{cccccccc}
\bar{\vf}(t-1) & \bar{\vf}(t-2) & \ldots  & \bar{\vf}(t-w)  \end{array} \right]
\; \vb \; \approx \; 
 \begin{array}{c}
\vp(t)  \end{array}\]
where $w$ is the window-size,  and $\bar{\vf}(\cdot)$ is either page views or both page views and transient scores.
After fitting $\vb$, the model predicts $\vp(t+1)$ as 
\[
\left[ \begin{array}{cccccccc}
\bar{\vf}(t) & \bar{\vf}(t-1) & \cdots  & \bar{\vf}(t-w+1)  \end{array} \right]
\; \vb
\]
We use the symmetric Mean Absolute Percentage Error (sMAPE)~\cite{ahmed2010empirical} measure to evaluate the prediction:
\[\text{sMAPE} = \frac{1}{|T|} \sum_{t=1}^{|T|} \frac{|p_t - \hat{p}_t|}{(p_t + \hat{p}_t)/2}.\]
This relative error measure averages all the relative prediction errors over all the time-steps.  We then average it over nodes.

We study two predictive modes.  The \emph{base model} uses only the time-series of page views  or tweets to predict the future page views or tweets. The \emph{dynamic teleportation model} uses both the transient scores with smooting and page views to predict the future page views (or tweets).

We evaluate these models for prediction on \textit{stationary} and \textit{non-stationary} time-series. 
Informally, a time-series is weakly stationary if it has properties (mean and covariance) similar to that of the time-shifted time-series. 
We consider the top and bottom 10,000 nodes from the difference ranking as nodes that are approximately non-stationary (volatile) and stationary (stable), respectively.  
Table~\ref{table:pv-preds} compares the predictions of the models across time for non-stationary and stationary prediction tasks. 
Our findings indicate that the Dynamic PageRank time-series provides valuable information for forecasting future tweet rates; however, it adds little (if any) accuracy in forecasting future page views on Wikipedia. 

\begin{table}
\caption{The ratio between the base model and the model with dynamic teleportation scores with $s=1,2,6,$ and $\infty$, for three smoothing parameters.  (Here, $s=\infty$ corresponds to solving the PageRank problem exactly for each change in teleportation.)  If this ratio is less than $1$, then the model with the dynamic teleportation scores improves the prediction performance.  We also distinguish between prediction problems with highly volatile nodes (non-stationary) and nodes with relatively stable behavior (stationary). The results show a much stronger benefit for Twitter than for Wikipedia}
\label{table:pv-preds}
\centering\scriptsize
\begin{tabularx}{\linewidth}{ ll XXXXX}
\toprule
\textbf{{Dataset}} & \textbf{Type} & $\theta$ & \multicolumn{4}{l}{\textbf{Error Ratio}} \\
\cmidrule{4-7}
 &  &  & \multicolumn{2}{l}{$s$ (timescale)} & & \\
 &  &  & \textbf{1} & \textbf{2} & \textbf{6} & $\mathbf{\infty}$ \\
\midrule
\textsc{twitter} & stationary & 
	\textbf{0.01}  & 0.635 & 0.929  &  0.913  &  0.996 \\ 
 & & 	\textbf{0.50}  & 0.636 & 0.735  &  0.854  &  0.939 \\ 
 & & 	\textbf{1.00}  & 0.522 & 0.562  &  0.710  &  0.963 \\ 
\addlinespace 
 & non-stationary & 
	 \textbf{0.01}  & 0.461 & 0.841  &  1.001  &  0.992 \\ 
 & & 	 \textbf{0.50}  & 0.261 & 0.608  &  0.585  &  0.929 \\ 
 & & 	 \textbf{1.00}  & 0.137 & 0.605  &  0.617  &  0.918 \\ 
\midrule
\textsc{wikipedia} & stationary & 
	\textbf{0.01}  & 0.978 & 0.991  &  0.989  &  0.978 \\ 
 & & 	\textbf{0.50}  & 1.140 & 1.130  &  1.004  &  0.990 \\ 
 & & 	\textbf{1.00}  & 1.084 & 0.976  &  1.010  &  0.990 \\ 
\addlinespace 
 & non-stationary & 
	 \textbf{0.01}  & 0.968 & 1.011  &  0.968  &  1.004 \\ 
 & & 	 \textbf{0.50}  & 1.218 & 0.994  &  1.030  &  1.031 \\ 
 & & 	 \textbf{1.00}  & 1.241 & 0.996  &  0.957  &  0.998 \\ 
\bottomrule
\end{tabularx} 
\end{table}

 
For Twitter, the dynamic teleportation model improves predictions the most with the  non-stationary nodes. The diffusion of activity captured by the model allows our model to detect, early on, when the external interest of vertices will change, before that change becomes apparent in the external interest of the vertices.
This is easiest to detect when there is a large sudden change  in external interest of a neighboring vertex.

 \begin{figure*}[t!]
\centering
    \subfigure[Temporal Patterns]
    	{\label{fig:centroids}\includegraphics[width=4.2in]{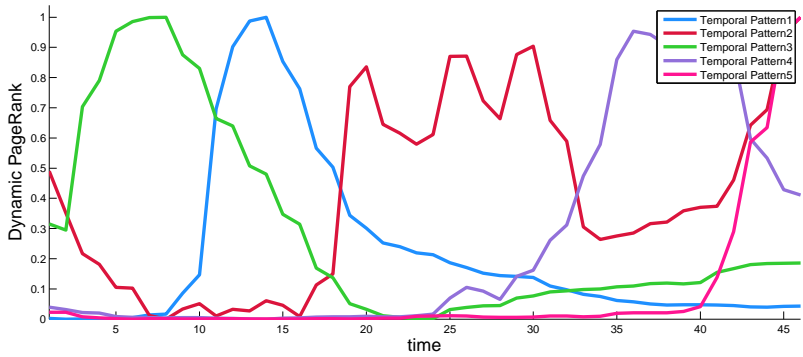}}
    \subfigure[Vertices with Similar Dynamics]
    	{\label{fig:similar-timeseries}\includegraphics[width=5.3in]{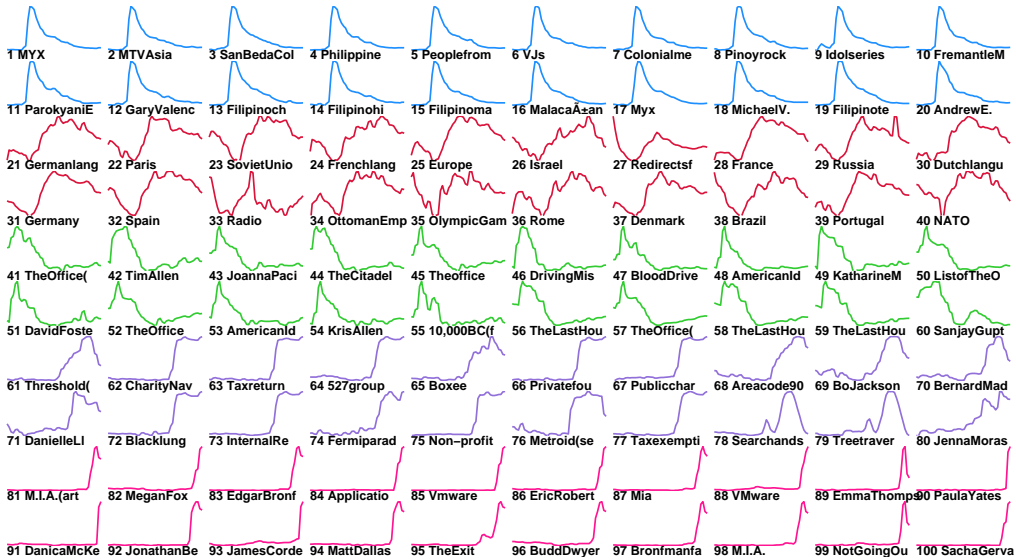}}
  	  \caption{Vertices with similar dynamical properties are grouped together. The visualization reveals the important dynamic patterns (spikes, trends) present from March 6th, 2009 in our large collection of time-series from Wikipedia.
 For each hour, we sample twice from the continuous function $\vx(t)$ and utilize these intermediate values in the clustering.
  	  } 	  
  \label{fig:dpr-clustering}
\end{figure*}

\subsection{Clustering transient score trends}
\label{sec:clustering}
Identifying vertices with similar time-series is important for modeling and understanding large collections of multivariate time-series.
We now group vertices according to their transient scores. By using the difference rank measure $\vd$ for $s=4$, we cluster the top 5,000 vertices using k-means with $k=5$, repeat the clustering 2,000 times, and take the minimum distance clustering identified. 

The cluster centroids are temporal patterns, and the main patterns in the dynamic PageRanks are visualized in \figref{centroids}.
Pattern 2 represents European-centric behavior, whereas the others correspond to spikes or unusual events occurring within the dynamic PageRank system.
\Figref{similar-timeseries} plots the 20 closest vertices matching the patterns above.
A few pages from the five groups are consistent with our previously discussed results from \figref{evolving-dpr-ts4}.
One such unusual event is related to the death of a famous musician/actor from the Philippines (see pages 1-20).
The pages from the third cluster (41-60) are related to ``American Idol'' and other TV shows/actors.
Also some of the pages from the fourth cluster relate to Bernard Madoff (63, 66, 67, 70, 73), six days before he plead guilty in the largest financial fraud in U.S. history.
This grouping reveals many of the standard patterns in time-series such as spikes and increasing/decreasing trends~\cite{Yang-2011-temporal-variation}.

\subsection{Towards causal link relationships}
\label{sec:granger}
In this section, we use Granger causality tests~\cite{granger1969investigating} on the collection of transient scores to attempt to understand which links are most important.  The Granger causality model, briefly described below, ought to identify a causal relationship between the time-series of any two vertices connected by a directed edge.  This is because \emph{there is a causal relationship} between their time-series in our dynamical system.  However, due to the impact of the time-dependent teleportation, only some of these links will be identified as causal.  We wish to investigate this smaller subset of links.

Intuitively, if a time-series $X$ causally affects another $Y$, then the past values of $X$ should be helpful in predicting the future values of $Y$, above what can be predicted based on the past values of $Y$ alone.  This is formalized as follows:
the error in predicting $\hat{y}_{t+s}$ from $y_t, y_{t-1}, \ldots$ should be larger than 
the error in predicting $\hat{y}_{t+s}$ from the joint data $y_t, y_{t_1}, \ldots, x_{t}, x_{t-1}, \ldots$ if $X$ causes $Y$. As our model, we chose to use the standard vector-autoregressive (VAR) model from econometrics~\cite{box2011time}. 
This is implemented in a Matlab code by~\citet{lesage1999applied}.
The standard $p$-lag VAR model takes the form:
\[
\left[ 
\begin{array}{c}
y_{t} \\
x_{t} \\
\end{array} 
\right]
= \vc + \sum_{i=1}^{p} \mathrm{\mathbf{M}}_i 
\left[ 
\begin{array}{c}
y_{t-i} \\
x_{t-i} \\
\end{array} \right] + \mathbf{e}_t\]
where $\vc$ is a vector of constants, $\mathbf{\mM}_i$ are the $n \times n$ coefficient (or autoregressive) mixing matrices and $\ve_t$ is the unobservable white-noise. For the results shown below, $p=2$.
We then use the standard F-test to determine significance.

In Table~\ref{table:causality-example}, we show the causal relationships identified among the out-links of the article \emph{Earthquake}.  Recall that there was a major earthquake in Australia during our time-window.  We wish to understand which of the out-links appeared to be sensitive to this large change in interest in Earthquake.
We use a significance cutoff of 0.01 and test for Granger causality among the time-series with $s=4$.

\begin{table}[t!]
\caption{Example of Causality in Wikipedia. We only consider pages with a $pval < 0.01$ as statistically significant. 
The page with values ``caused'' by Earthquake represent ideas related to earthquakes. 
All pages below are significant with $pval < 0.01$.}
\label{table:causality-example}
\centering\small
\begin{tabularx}{3in}{ r X l}
\toprule
 \textbf{{Earthquake Granger causes}} & & \textbf{{p-value}} \\
\midrule
Seismic hazard & & $0.003535 $ \\
Extensional tectonics & & $0.003033$ \\
Landslide dam & & $0.002406$  \\
Earthquake preparedness & & $0.001157$  \\
Richter magnitude scale & & $0.000584$  \\
Fault (geology) & & $0.000437$  \\
Aseismic creep & & $0.000419$  \\
Seismometer & & $0.000284$  \\
Epicenter & & $0.000020$  \\
Seismology & & $0.000001$  \\
\bottomrule
\end{tabularx}
\end{table}

\section{Conclusion}
PageRank is one of the most widely used network centrality measures. Our dynamical system reformulation of PageRank permits us to incorporate time-dependent teleportation in a relatively seamless manner.  Based on the results presented here, we believe this is an interesting variation on the PageRank model.  For instance, we can analyze certain choices of oscillating teleportation functions (Lemma~\ref{lem:fluctuate}).  Our empirical results show that the maximum change in the transient rank values identifies interesting sets of pages.  
Furthermore, this method is simple to implement in an online setting using either a forward Euler or Runge-Kutta integrator for the dynamical system.  We hope that it might find a use in online monitoring systems.

One important direction for future work is to treat the inverse problem.  That is, suppose that the observed page views reflect the behavior of these random surfers.  Formally, suppose that we equate page views with samples of $\vx(t)$.  Then, the goal would be to find $\vv(t)$ that produces this $\vx(t)$.  This may not be a problem for websites such as Wikipedia, due to our argument that the majority of page views reflect search engine traffic.  But for many other cases, we suspect that $\vx(t)$ may be much easier to observe.

\bibliographystyle{plainnat}
\bibliography{dpr_im}

\begin{thebibliography}{41}
\providecommand{\natexlab}[1]{#1}
\providecommand{\url}[1]{\texttt{#1}}
\expandafter\ifx\csname urlstyle\endcsname\relax
  \providecommand{\doi}[1]{doi: #1}\else
  \providecommand{\doi}{doi: \begingroup \urlstyle{rm}\Url}\fi

\bibitem[Abiteboul et~al.(2003)Abiteboul, Preda, and
  Cobena]{abiteboul2003adaptive}
S.~Abiteboul, M.~Preda, and G.~Cobena.
\newblock Adaptive on-line page importance computation.
\newblock In \emph{WWW}, pages 280--290. ACM, 2003.

\bibitem[Ahmed et~al.(2010)Ahmed, Atiya, El~Gayar, and
  El-Shishiny]{ahmed2010empirical}
N.K. Ahmed, A.F. Atiya, N.~El~Gayar, and H.~El-Shishiny.
\newblock An empirical comparison of machine learning models for time series
  forecasting.
\newblock \emph{Econ. Rev.}, 29\penalty0 (5-6):\penalty0 594--621, 2010.

\bibitem[Andersen et~al.(2006)Andersen, Chung, and Lang]{andersen2006-local}
Reid Andersen, Fan Chung, and Kevin Lang.
\newblock Local graph partitioning using {PageRank} vectors.
\newblock In \emph{Proceedings of the 47th Annual IEEE Symposium on Foundations
  of Computer Science}, 2006.
\newblock URL \url{http://www.math.ucsd.edu/~fan/wp/localpartition.pdf}.

\bibitem[Bagrow et~al.(2011)Bagrow, Wang, and
  Barab{\'a}si]{bagrow2011collective}
J.P. Bagrow, D.~Wang, and A.L. Barab{\'a}si.
\newblock Collective response of human populations to large-scale emergencies.
\newblock \emph{PloS one}, 6\penalty0 (3):\penalty0 e17680, 2011.

\bibitem[Bahmani et~al.(2012)Bahmani, Kumar, Mahdian, and
  Upfal]{bahmani2012pagerank}
B.~Bahmani, R.~Kumar, M.~Mahdian, and E.~Upfal.
\newblock {PageRank} on an evolving graph.
\newblock In \emph{Proceedings of the 18th ACM SIGKDD international conference
  on Knowledge discovery and data mining}, pages 24--32. ACM, 2012.

\bibitem[Becchetti et~al.(2008)Becchetti, Castillo, Donato, Baeza-Yates, and
  Leonardi]{becchetti2008-spam}
Luca Becchetti, Carlos Castillo, Debora Donato, Ricardo Baeza-Yates, and
  Stefano Leonardi.
\newblock Link analysis for web spam detection.
\newblock \emph{ACM Trans. Web}, 2\penalty0 (1):\penalty0 1--42, February 2008.
\newblock ISSN 1559-1131.
\newblock \doi{10.1145/1326561.1326563}.

\bibitem[Berman et~al.(1989)Berman, Neumann, and
  Stern]{Berman-1989-nonnegative}
Abraham Berman, Michael Neumann, and Ronald~J. Stern.
\newblock \emph{Nonnegative Matrices in Dynamic Systems}.
\newblock Wiley, 1989.

\bibitem[Bianchini et~al.(2005)Bianchini, Gori, and
  Scarselli]{bianchini2005-inside-pagerank}
M.~Bianchini, M.~Gori, and F.~Scarselli.
\newblock Inside {PageRank}.
\newblock \emph{ACM Transactions on Internet Technologies}, 5\penalty0
  (1):\penalty0 92--128, 2005.
\newblock ISSN 1533-5399.
\newblock \doi{10.1145/1052934.1052938}.

\bibitem[Boldi(2005)]{Boldi05totalrank}
Paolo Boldi.
\newblock {TotalRank}: Ranking without damping.
\newblock In \emph{WWW}, pages 898--899, 2005.

\bibitem[Boldi et~al.(2005)Boldi, Santini, and
  Vigna]{boldi2005-incremental-pagerank}
Paolo Boldi, Massimo Santini, and Sebastiano Vigna.
\newblock Paradoxical effects in {PageRank} incremental computations.
\newblock \emph{Internet Mathematics}, 2\penalty0 (2):\penalty0 387--404, 2005.

\bibitem[Boldi et~al.(2007)Boldi, Posenato, Santini, and
  Vigna]{boldi2007-traps}
Paolo Boldi, Roberto Posenato, Massimo Santini, and Sebastiano Vigna.
\newblock Traps and pitfalls of topic-biased {PageRank}.
\newblock In \emph{WAW2006, Fourth International Workshop on Algorithms and
  Models for the Web-Graph}, LNCS, pages 107--116. Springer-Verlag, 2007.
\newblock \doi{10.1007/978-3-540-78808-9_10}.

\bibitem[Box et~al.(2011)Box, Jenkins, and Reinsel]{box2011time}
G.E.P. Box, G.M. Jenkins, and G.C. Reinsel.
\newblock \emph{Time series analysis: forecasting and control}, volume 734.
\newblock Wiley, 2011.

\bibitem[Chien et~al.(2004)Chien, Dwork, Kumar, Simon, and
  Sivakumar]{chien2004link}
S.~Chien, C.~Dwork, R.~Kumar, D.R. Simon, and D.~Sivakumar.
\newblock Link evolution: Analysis and algorithms.
\newblock \emph{Internet Mathematics}, 1\penalty0 (3):\penalty0 277--304, 2004.

\bibitem[Constantine and Gleich(2010)]{constantine2010-rapr}
Paul~G. Constantine and David~F. Gleich.
\newblock Random alpha {PageRank}.
\newblock \emph{Internet Mathematics}, 6\penalty0 (2):\penalty0 189--236,
  September 2010.
\newblock \doi{10.1080/15427951.2009.10129185}.

\bibitem[Das~Sarma et~al.(2008)Das~Sarma, Gollapudi, and
  Panigrahy]{das2008estimating}
A.~Das~Sarma, S.~Gollapudi, and R.~Panigrahy.
\newblock Estimating {PageRank} on graph streams.
\newblock In \emph{SIGMOD}, pages 69--78. ACM, 2008.

\bibitem[Dunlavy et~al.(2011)Dunlavy, Kolda, and Acar]{Dunlavy-2011-temporal}
Daniel~M. Dunlavy, Tamara~G. Kolda, and Evrim Acar.
\newblock Temporal link prediction using matrix and tensor factorizations.
\newblock \emph{TKDD}, 5\penalty0 (2):\penalty0 10:1--10:27, February 2011.
\newblock ISSN 1556-4681.
\newblock \doi{10.1145/1921632.1921636}.

\bibitem[Embree and Lehoucq(2009)]{Embree-2009-dynamical}
Mark Embree and Richard~B. Lehoucq.
\newblock Dynamical systems and non-hermitian iterative eigensolvers.
\newblock \emph{SIAM Journal on Numerical Analysis}, 47\penalty0 (2):\penalty0
  1445--1473, 2009.
\newblock \doi{10.1137/07070187X}.

\bibitem[Gleich et~al.(2007)Gleich, Glynn, Golub, and Greif]{gleich2007three}
D.~Gleich, P.~Glynn, G.~Golub, and C.~Greif.
\newblock Three results on the {PageRank} vector: eigenstructure, sensitivity,
  and the derivative.
\newblock \emph{Web Information Retrieval and Linear Algebra Algorithms}, 2007.

\bibitem[Gleich et~al.(2010)Gleich, Constantine, Flaxman, and
  Gunawardana]{gleich2010tracking}
D.F. Gleich, P.G. Constantine, A.D. Flaxman, and A.~Gunawardana.
\newblock Tracking the random surfer: empirically measured teleportation
  parameters in {PageRank}.
\newblock In \emph{WWW}, pages 381--390. ACM, 2010.

\bibitem[Granger(1969)]{granger1969investigating}
C.W.J. Granger.
\newblock Investigating causal relations by econometric models and
  cross-spectral methods.
\newblock \emph{Econometrica: Journal of the Econometric Society}, pages
  424--438, 1969.

\bibitem[Grindrod et~al.(2011)Grindrod, Parsons, Higham, and
  Estrada]{grindrod2011communicability}
P.~Grindrod, M.C. Parsons, D.J. Higham, and E.~Estrada.
\newblock Communicability across evolving networks.
\newblock \emph{Physical Review E}, 83\penalty0 (4):\penalty0 046120, 2011.

\bibitem[Gy{\"o}ngyi et~al.(2004)Gy{\"o}ngyi, Garcia-Molina, and
  Pedersen]{gyongyi2004-trustrank}
Zolt{\'a}n Gy{\"o}ngyi, Hector Garcia-Molina, and Jan Pedersen.
\newblock Combating web spam with {TrustRank}.
\newblock In \emph{Proceedings of the 30th International Very Large Database
  Conference}, Toronto, Canada, 2004.
\newblock ISBN 0-12-088469-0.
\newblock URL
  \url{http://i.stanford.edu/~zoltan/publications/gyongyi2004combating.pdf}.

\bibitem[Horn and Serra-Capizzano(2007)]{horn2008-parametric-google}
Roger~A. Horn and Stefano Serra-Capizzano.
\newblock A general setting for the parametric {Google} matrix.
\newblock \emph{Internet Mathematics}, 3\penalty0 (4):\penalty0 385--411, March
  2007.
\newblock URL \url{http://projecteuclid.org/euclid.im/1227025007}.

\bibitem[Kutta(1901)]{kutta1901beitrag}
W.~Kutta.
\newblock Beitrag zur n{\"a}herungweisen integration totaler
  differentialgleichungen.
\newblock 1901.

\bibitem[Langville and Meyer(2004)]{langville04-iad}
Amy~N. Langville and Carl~D. Meyer.
\newblock Updating {PageRank} with iterative aggregation.
\newblock In \emph{WWW}, pages 392--393, 2004.

\bibitem[Langville and Meyer(2006)]{langville2006-book}
Amy~N. Langville and Carl~D. Meyer.
\newblock \emph{{Google}'s {PageRank} and Beyond: The Science of Search Engine
  Rankings}.
\newblock Princeton University Press, 2006.
\newblock ISBN 978-0-691-12202-1.

\bibitem[LeSage(1999)]{lesage1999applied}
J.P. LeSage.
\newblock Applied econometrics using {MATLAB}.
\newblock \emph{Manuscript, Dept. of Economics, University of Toronto}, 1999.

\bibitem[O'Madadhain and Smyth(2005)]{o2005eventrank}
J.~O'Madadhain and P.~Smyth.
\newblock {EventRank}: A framework for ranking time-varying networks.
\newblock In \emph{LinkKDD}, pages 9--16. ACM, 2005.

\bibitem[Page et~al.(1998)Page, Brin, Motwani, and Winograd]{page1998pagerank}
L.~Page, S.~Brin, R.~Motwani, and T.~Winograd.
\newblock {The {PageRank} citation ranking: Bringing order to the web}.
\newblock 1998.

\bibitem[Ratkiewicz et~al.(2010)Ratkiewicz, Fortunato, Flammini, Menczer, and
  Vespignani]{ratkiewicz2010characterizing}
J.~Ratkiewicz, S.~Fortunato, A.~Flammini, F.~Menczer, and A.~Vespignani.
\newblock Characterizing and modeling the dynamics of online popularity.
\newblock \emph{PRL}, 105\penalty0 (15):\penalty0 158701, 2010.

\bibitem[Rossi and Gleich(2012)]{Rossi-2012-dynamicpr}
Ryan~A. Rossi and David~F. Gleich.
\newblock Dynamic pagerank using evolving teleportation.
\newblock In Anthony Bonato and Jeannette Janssen, editors, \emph{Algorithms
  and Models for the Web Graph}, volume 7323 of \emph{Lecture Notes in Computer
  Science}, pages 126--137. Springer Berlin Heidelberg, 2012.
\newblock ISBN 978-3-642-30540-5.
\newblock \doi{10.1007/978-3-642-30541-2_10}.

\bibitem[Runge(1895)]{runge1895numerische}
C.~Runge.
\newblock {\"U}ber die numerische aufl{\"o}sung von differentialgleichungen.
\newblock \emph{Mathematische Annalen}, 46\penalty0 (2):\penalty0 167--178,
  1895.

\bibitem[Shampine and Reichelt(1997)]{Shampine-1997-ODEs}
Lawrence~F. Shampine and Mark~W. Reichelt.
\newblock The {MATLAB} {ODE} suite.
\newblock \emph{SIAM Journal on Scientific Computing}, 18\penalty0
  (1):\penalty0 1--22, 1997.
\newblock \doi{10.1137/S1064827594276424}.

\bibitem[Singh et~al.(2007)Singh, Xu, and Berger]{singh2007-matching-topology}
Rohit Singh, Jinbo Xu, and Bonnie Berger.
\newblock Pairwise global alignment of protein interaction networks by matching
  neighborhood topology.
\newblock In \emph{Proceedings of the 11th Annual International Conference on
  Research in Computational Molecular Biology (RECOMB)}, volume 4453 of
  \emph{Lecture Notes in Computer Science}, pages 16--31, Oakland, CA, 2007.
  Springer Berlin / Heidelberg.
\newblock \doi{10.1007/978-3-540-71681-5_2}.

\bibitem[Sun et~al.(2006)Sun, Tao, and Faloutsos]{Sun-2006-tensor}
Jimeng Sun, Dacheng Tao, and Christos Faloutsos.
\newblock Beyond streams and graphs: dynamic tensor analysis.
\newblock In \emph{SIGKDD}, KDD '06, pages 374--383, New York, NY, USA, 2006.
  ACM.
\newblock ISBN 1-59593-339-5.
\newblock \doi{10.1145/1150402.1150445}.

\bibitem[Tong et~al.(2006)Tong, Faloutsos, and
  Pan]{tong2006-random-walk-restart}
Hanghang Tong, Christos Faloutsos, and Jia-Yu Pan.
\newblock Fast random walk with restart and its applications.
\newblock In \emph{ICDM '06: Proceedings of the Sixth International Conference
  on Data Mining}, pages 613--622, Washington, DC, USA, 2006. IEEE Computer
  Society.
\newblock ISBN 0-7695-2701-9.
\newblock \doi{10.1109/ICDM.2006.70}.

\bibitem[Tsaparas(2004)]{Tsaparas-2004-dynamical-systems}
Panayiotis Tsaparas.
\newblock Using non-linear dynamical systems for web searching and ranking.
\newblock In \emph{Proceedings of the twenty-third ACM SIGMOD-SIGACT-SIGART
  symposium on Principles of database systems}, PODS '04, pages 59--70, New
  York, NY, USA, 2004. ACM.
\newblock ISBN 158113858X.
\newblock \doi{10.1145/1055558.1055569}.

\bibitem[Various(2009)]{wikipedia2009}
Various.
\newblock Wikipedia database dump, 2009.
\newblock Version from 2009-03-06. {\scriptsize
  \url{http://en.wikipedia.org/wiki/Wikipedia:Database_download}}.

\bibitem[Various(2011)]{pageviews2009}
Various.
\newblock Wikipedia pageviews, 2011.
\newblock Accessed in 2011. {\scriptsize
  \url{http://dumps.wikimedia.org/other/pagecounts-raw/}}.

\bibitem[Wills and Ipsen(2009)]{wills2008-ordinal}
Rebecca~S. Wills and Ilse C.~F. Ipsen.
\newblock Ordinal ranking for {G}oogle's {P}age{R}ank.
\newblock \emph{SIAM Journal on Matrix Analysis and Applications}, 30:\penalty0
  1677--1696, January 2009.
\newblock \doi{10.1137/070698129}.

\bibitem[Yang and Leskovec(2011)]{Yang-2011-temporal-variation}
Jaewon Yang and Jure Leskovec.
\newblock Patterns of temporal variation in online media.
\newblock In \emph{Proceedings of the fourth ACM international conference on
  Web search and data mining}, WSDM '11, pages 177--186, New York, NY, USA,
  2011. ACM.
\newblock ISBN 978-1-4503-0493-1.
\newblock \doi{10.1145/1935826.1935863}.

\end{thebibliography}

\end{document}